\documentclass[a4paper,cleveref,autoref,thm-restate]{lipics-v2019}

\title{On Decidability of Time-Bounded Reachability in CTMDPs} 

\titlerunning{On Decidability of Time-Bounded Reachability in CTMDPs} 

\author{Rupak Majumdar}{MPI-SWS, Kaiserslautern, Germany}{rupak@mpi-sws.org}{https://orcid.org/0000-0003-2136-0542}{}
\author{Mahmoud Salamati}{MPI-SWS, Kaiserslautern, Germany}{msalamati@mpi-sws.org}{https://orcid.org/0000-0003-3790-3935}{}

\author{Sadegh Soudjani}{Newcastle University, Newcastle upon Tyne, United Kingdom}{sadegh.soudjani@newcastle.ac.uk}{https://orcid.org/0000-0003-1922-6678}{}

\authorrunning{R. Majumdar, M. Salamati, and S. Soudjani} 

\Copyright{Rupak Majumdar, Mahmoud Salamati, and Sadegh Soudjani} 



\category{Track B: Automata, Logic, Semantics, and Theory of Programming}

\supplement{}

\acknowledgements{This research was funded in part by the Deutsche Forschungsgemeinschaft project 389792660-TRR 248 and by the European Research Council under the Grant Agreement 610150 (ERC Synergy Grant ImPACT).
We thank Jo\"{e}l Ouaknine, James Worrell, and Joost-Pieter Katoen for discussions and pointers.}

\nolinenumbers 

\EventEditors{}
\EventNoEds{0}
\EventLongTitle{47th International Colloquium on Automata, Languages and Programming (ICALP 2020)}
\EventShortTitle{ICALP2020}
\EventAcronym{ICALP}
\EventYear{2020}
\EventDate{July 8--12, 2020}
\EventLocation{Beijing, Chaina}
\EventLogo{}
\SeriesVolume{}

\usepackage{etex}
\usepackage{wrapfig}
\usepackage{graphicx}
\usepackage{amssymb,amsmath,amsthm}
\usepackage[mathscr]{eucal}
\usepackage{verbatim}
\usepackage{comment}
\usepackage{mathtools}
\usepackage{amsfonts}
\usepackage{color}
\usepackage{graphicx}
\usepackage{tikz}
\usepackage{xspace} 
\usepackage{pgf}
\usetikzlibrary{arrows,automata}
\usepackage[usenames,dvipsnames]{pstricks}
\usepackage{epsfig}
\usepackage{pst-grad} 
\usepackage{pst-plot} 
\usepackage{pstricks-add}
\usepackage[linesnumbered,ruled]{algorithm2e}
\usepackage{hyperref}

\newcommand{\Mahmoud}[1]{\textcolor{blue}{#1}}
\newcommand{\RM}[1]{\textcolor{red}{#1}}

\newcommand{\dv}{\mathbf{d}\xspace}
\newcommand{\Qb}{\mathbf{Q}\xspace}
\newcommand{\reach}{\textbf{reach}\xspace}

\newcommand{\Stationary}{\Pi_{\mathfrak{st}}}

\newtheorem{problem}{Problem}

\def\reals{\mathbb{R}}
\def\nats{\mathbb{N}}

\def\rats{\mathbb{Q}}

\def\theory{{\mathbb{R}_{\mathsf{MW}}}}

\def\set#1{{\{ #1 \}}}
\def\tuple#1{{\langle #1 \rangle}}

\def\Prob{{\mathrm{Prob}}}

\def\M{{\mathcal{M}}}
\def\N{{\mathcal{N}}}

\def\good{\mathbf{good}}
\def\bad{\mathbf{bad}}

\EventEditors{Artur Czumaj, Anuj Dawar, and Emanuela Merelli}
\EventNoEds{3}
\EventLongTitle{47th International Colloquium on Automata, Languages, and Programming (ICALP 2020)}
\EventShortTitle{ICALP 2020}
\EventAcronym{ICALP}
\EventYear{2020}
\EventDate{July 8--11, 2020}
\EventLocation{Saarbrücken, Germany (virtual conference)}
\EventLogo{}
\SeriesVolume{168}
\ArticleNo{133}

\begin{document}
\maketitle

\begin{abstract}
We consider the time-bounded reachability problem for continuous-time Markov decision processes.
We show that the problem is decidable subject to Schanuel's conjecture.
Our decision procedure relies on the structure of optimal policies and the conditional decidability 
(under Schanuel's conjecture) of
the theory of reals extended with exponential and trigonometric functions over bounded domains.
We further show that any unconditional decidability result would imply unconditional decidability of
the bounded continuous Skolem problem, or equivalently, the problem of checking if an exponential
polynomial has a non-tangential zero in a bounded interval.
We note that the latter problems are also decidable subject to Schanuel's conjecture but
finding unconditional decision procedures remain longstanding open problems.
\end{abstract}

\clearpage


\section{Introduction}
\label{sec:Intro}

Continuous-time Markov decision processes (CTMDPs) are a widely used model for continuous-time
systems which exhibit both stochastic and non-deterministic choice.
A CTMDP consists of a finite set of states, a finite set of actions, and for each action,
a transition rate matrix that determines the rate (in an exponential distribution in continuous time)
to go from one state to the next when the action is chosen. 
A \emph{policy} for a CTMDP maps a timed execution path to state-dependent actions.
Given a fixed policy, a CTMDP determines a stochastic process in continuous time, where
the rate matrix determines the distribution of switches.

A fundamental decision problem for CTMDPs is the 
\emph{time-bounded reachability problem}, which asks, given a CTMDP $\M$ with a designated ``$\good$'' state,
a time bound $B$, and a rational vector $r$, whether there exists a policy that 
controls the Markov decision process such that the probability of reaching
the good state from state $s$ within time bound $B$ is at least $r(s)$.
The time-bounded reachability problem is at the core of model checking CTMDPs with respect to 
stochastic temporal logics \cite{Baier2005} and has been extensively 
studied \cite{Buchholz2011,Neuhausser2010,Wolovick2006,Neuhuer2009,Hermanns:2011}.

Existing papers either consider \emph{time-abstract} policies 
\cite{Baier2005,Rabe2013,Brzdil2013,Wolovick2006,Neuhuer2009}
or focus on numerical approximation schemes \cite{Buchholz2011,Neuhausser2010,MedinaAyala2014,Fearnley2016,Hermanns:2011,Mahmoud2020_TOMPECS}.
However, policies that depend on time are strictly more powerful and the \emph{decision problem} has remained open.
For the special case of continuous-time Markov chains (CTMCs), where each state has a unique
action, the time-bounded reachability problem is decidable \cite{Aziz2000}.
The proof uses tools from transcendental number theory, specifically, the Lindemann-Weierstrass theorem.
One might expect that a similar argument might be used to show decidability for CTMDPs as well.

In this paper, we show \emph{conditional} decidability.
Our result uses, like several other conditional results on dynamical systems, Schanuel's conjecture from
transcendental number theory (see, e.g., \cite{Lang1966}).
Our proof has the following ingredients.
First, we use the fact that the optimal policy for the time-bounded reachability problem is a timed,
\emph{piecewise constant} function with a \emph{finite} number of switches \cite{Miller_68,Neuhuer2010,Rabe:2011}.
We show that each switch point of an optimal policy corresponds to a non-tangential zero of an associated linear dynamical system.
Second, we use the result of Macintyre and Wilkie \cite{MACINTYRE2008,Macintyre2016} 
that Schanuel's conjecture implies the 
decidability of the real-closed field together with the exponential, sine, and cosine functions over a bounded domain.
The existence of non-tangential zeros of linear dynamical systems can be encoded in this theory.
Third, for each natural number $k\in \nats$,
we write a sentence in this theory whose validity implies there is an optimal strategy with exactly $k$ switch points.
By enumerating over $k$, we find the exact number of switches in an optimal strategy.
Finally, we write another sentence in the theory that checks if the reachability probability attained by 
(an encoding of) the optimal policy is greater than the given bound.

We also study the related decision problem whether there is a \emph{stationary} (i.e., time independent) optimal policy.
We show that there is a ``direct'' conditional decision procedure for this problem based on Schanuel's conjecture and
recent results on zeros of exponential polynomials \cite{joel2016}, which avoids the result of Macintyre and Wilkie.

At the same time, we show that an \emph{unconditional} decidability result is likely to be very difficult.
We show that the bounded continuous-time Skolem problem \cite{Bell2010,joel2016} reduces to checking if there is an optimal
stationary policy in the time-bounded CTMDP problem. 
The bounded continuous Skolem problem is a long-standing open problem about linear dynamical systems
\cite{joel2016,Bell2010}; it asks if a linear dynamical system in continuous 
time has a non-tangential zero in a bounded interval.
%
Our reduction, in essence, demonstrates that CTMDPs can ``simulate'' any linear dynamical system: a non-tangential 
zero in the dynamics corresponds to a policy switch point in the simulating CTMDP.

Our result is in the same spirit as several recent results providing conditional decision procedures, based
on Schanuel's conjecture, or hardness results, based on variants of the Skolem problem, for problems on probabilistic systems.
For example, Daviaud~et~al.~\cite{ben2018} showed conditional decidability of subcases of 
the containment problem for probabilistic automata subject to the conditional decidability of the theory of
real closed fields with the exponential function \cite{Macintyre1996,Wilkie1997}.
For lower bounds, Akshay~et~al. \cite{Akshay2015} showed a reduction from the (unbounded, discrete) Skolem problem to 
reachability on discrete time Markov chains and
Piribauer and Baier \cite{Baier2020} show that the positivity problem in discrete time can be reduced 
into several decision problems corresponding to optimization tasks over discrete time MDPs. 

In summary, we summarize our contribution as the following theorem.

\begin{theorem}
\label{th:main-theorem}
(1)~The time-bounded reachability problem for CTMDPs is decidable assuming Schanuel's conjecture.
(2)~Whether the time-bounded reachability problem has a stationary optimal policy is decidable assuming Schanuel's conjecture.
(3)~The bounded continuous Skolem problem reduces to checking if the time-bounded reachability problem has a stationary optimal policy. 
\end{theorem}


\section{Continuous Time Markov Decision Processes}
\label{sec:preliminaries}

\begin{definition}
	A continuous-time Markov decision process (CTMDP) is a tuple $\M = (S,\mathcal D,\Qb)$ where
	\begin{itemize}
	\item $S = \{1,2,\ldots,\mathfrak n\}$ is a finite set of states for some $\mathfrak{n} > 0$;
	\item a set $\mathcal D = \prod_{s=1}^{\mathfrak n}\mathcal D_s$ of \emph{decision vectors}, 
              where $\mathcal D_s$ is a finite set of \emph{actions} that can be taken in state $s\in S$;
	\item $\Qb$ is a $\mathcal D$-indexed family of $\mathfrak n\times \mathfrak n$ \emph{generator matrices};
	      we write $\Qb^{\dv}$ for the generator matrix corresponding to
              the decision vector $\dv\in \mathcal D$.
              The entry $\Qb^{\dv}(s,s')\ge 0$ for $s'\neq s$ gives the \emph{rate of transition} from state 
              $s$ to state $s'$ under action $\dv(s)$, and $\Qb^{\dv}(s,s')$ is independent of elements of $\dv$ except $\dv(s)$.
              The entry $\Qb^{\dv}(s,s) = -\sum_{s'\ne s} \Qb^{\dv}(s,s')$.
	\end{itemize} 
\end{definition}

A CTMDP $\M = (S,\mathcal D,\Qb)$ with $|\mathcal D|=1$, i.e., when only a unique action can be taken in each state, 
is called a \emph{continuous-time Markov chain} (CTMC) and is simply denoted by the tuple $(S,\Qb)$, and with abuse of notation,
we also write $\Qb$ for the unique generator matrix. 
The CTMDP $\M$ reduces to a CTMC whenever a decision vector $\dv$ is fixed for all time on the CTMDP.  

Intuitively, $\Qb^{\dv}(s,s')>0$ indicates that by fixing a decision vector $\dv$, a transition from $s$ to $s'$ is possible and that the timing 
of the transition is exponentially distributed with rate $\Qb^{\dv}(s,s')$. 
If there are several states $s'$ such that $\Qb^{\dv}(s,s') > 0$, more than one transition will be possible. 
For each decision vector $\dv\in \mathcal D$ and any $s\in S$,
the total rate of taking an outgoing transition from state $s$ when $\dv$ is fixed is given by 
$E_{\dv}(s)=\sum_{s'\ne s} \Qb^{\dv}(s,s')$, 
By fixing this decision vector $\dv$, a transition from a state $s$ into $s'$ occurs within time $t$ with probability
\begin{equation*}
\mathbf{P}(s,s',t)=\frac{\Qb^{\dv}(s,s')}{E_{\dv}(s)}.(1-e^{-E_{\dv}(s)t}),\quad t\ge 0.
\end{equation*}
Intuitively, $1-e^{-E_{\dv}(s)t}$ is the probability of taking an outgoing transition at $s$ within time $t$ 
(exponentially distributed with rate $E_{\dv}(s)$) and 
$\Qb^{\dv}(s,s')/E_{\dv}(s)$ is the probability of taking transition to $s'$ among possible next states at $s$. 
Thus, the total probability of moving from $s$ to $s'$ under the decision $\dv$ in one transition, written $\mathbf{P}_{\dv}(s,s')$ is $\Qb^{\dv}(s,s')/E_{\dv}(s)$.
A state $s\in S$ is called \emph{absorbing} if and only if $\Qb^{\dv}(s,s') = 0$ for all $s'\in S$ and all 
decision vectors $\dv\in \mathcal D$.
For an absorbing state, we have $E_{\dv}(s) = 0$ for any decision vector $\dv$ and no transitions are enabled. 
The initial state of a CTMDP is either fixed deterministically or selected randomly according to a probability distribution $\alpha$ over the set of states $S$.

Consider a time interval $[0,B]$ with time bound $B>0$. 
Let $\varOmega$ denote the set of all right-continuous step functions $f:[0,B]\rightarrow S$, 
i.e., there are time points $t_0=0<t_1<t_2<\ldots<t_{m}=B$ such that $f(t') = f(t'')$ for all $t',t''\in[t_i,t_{i+1})$ for all $i\in\{0,1,\ldots,m-1\}$. 
Let $\mathcal F$ denote the sigma-algebra of the \emph{cylinder sets}
\begin{equation}
\label{eq:cylinder}
\textsf{Cyl}(s_0,I_0,\ldots,I_{m-1}, s_m): = \{f\in\varOmega\,|\, \forall 0\le i\le m \,\cdot\, f(t_i) = s_i \,\wedge\, i<m \Rightarrow (t_{i+1}-t_i)\in I_i\}.
\end{equation}
for all $m$, $s_i\in S$ and non-empty time intervals $I_0,I_1,\ldots,I_{m-1}\subset [0,B]$.
\begin{definition}
\label{def:policy}
A \emph{policy} $\pi$ is a function from $[0,B]$ into $\mathcal D$, which is assumed to be Lebesgue measurable. 
Any policy gives a decision vector $\pi_t\in \mathcal D$ at time $t$ such that the action $\pi_t(s)$ is taken when the CTMDP is at 
state $s$ at time $t$. The set of all such polices is denoted by $\Pi_B$. 
\end{definition}
Any policy $\pi$ together with an initial distribution $\alpha$ induces the probability space 
$(\varOmega,\mathcal F, \mathbf P_{\alpha}^\pi)$. 
If the initial distribution is chosen deterministically as $s\in S$, we denote the probability measure by $\mathbf P_{s}^\pi$ instead of $\mathbf P_{\alpha}^\pi$.

A policy $\pi:[0,B]\rightarrow\mathcal D$ is \emph{piecewise constant} if there exist a number $m\in \nats$ and time points
$t_0=0<t_1<t_2<\ldots<t_{m}=B$ such that $\pi_{t'} = \pi_{t''}$ for all $t',t''\in(t_i,t_{i+1}]$ and all $i\in\{0,1,\ldots,m-1\}$. 
The policy is \emph{stationary} if $m=1$. We denote the class of stationary policies by $\Stationary$; observe
that a stationary policy is given by a fixed decision vector, so $\Stationary$ is isomorphic with the set of
decision vectors $\mathcal{D}$.
In particular, it is a finite set.

\begin{remark}
The policies in Def.~\ref{def:policy} are called \emph{timed positional} policies since the action is selected deterministically as a function of time 
and the state of the CTMDP at that time. 
A stationary policy is only positional since the selected action is independent of time.
\end{remark}

\begin{problem}
\label{prob:prob_main}
Consider a CTMDP $\M = (\set{1,\ldots,n} \uplus\set{\good},\mathcal D, \Qb)$ with a distinguished absorbing state named $\good$ and
a time bound $B>0$. Define the event
\begin{equation}
\label{eq:reach_event}
\reach := \cup\{f\in\varOmega\mid f(t)=\good\text{ for some } t\in[0,B]\}.
\end{equation}
The \emph{time-bounded reachability problem} asks if for a rational vector $r \in [0,1]^n$, we have
	\begin{equation*}
	\label{eq:reach} 
		\sup_{\pi\in\Pi_B} \mathbf P_s^\pi(\reach)>r(s), \quad \text{ for all } s\in \set{1,\ldots,n}.
	\end{equation*}
\end{problem}
The event $\reach$ defined in \eqref{eq:reach_event} is written as a union of an uncountable number of functions but it is measurable in the probability space $(\varOmega,\mathcal F, \mathbf P_{\alpha}^\pi)$ for any $\alpha$. Since the state space is finite, $\reach$ can be written as a countable union of cylinder sets in the form of \eqref{eq:cylinder} by taking all the time intervals to be $[0,B]$ and enumerating over all possible sequence of states (which is countable) \cite{BK08}.   

A policy $\pi^* \in \Pi_B$ is \emph{optimal} if $P_s^{\pi^*}(\reach) = \sup_{\pi\in\Pi_B} \mathbf P_s^\pi(\reach)$.
Note that there are more general classes of policies that may depend also on the history of the states in the previous time points and which map the history
to a distribution over $\mathcal D$.
It is shown that piecewise constant timed positional policies are sufficient for the optimal reachability probability \cite{Miller_68,Neuhuer2010,Rabe:2011}.
That is, if there is an optimal policy from the larger class of policies, there is already one from the class of piecewise constant, timed, positional policies.

A closely related problem is the existence of \emph{stationary} optimal policies; here, it is possible that the optimal
stationary policy performs strictly worse than an optimal policy.

\begin{problem}
\label{prob:prob_switch}
Consider a CTMDP $\M = (\set{1,\ldots,n} \uplus\set{\good},\mathcal D, \Qb)$
and a time bound $B>0$. Decide whether there is an optimal policy $\pi^*$ that is stationary, namely
	\begin{equation*}
	\label{eq:reach_stationary} 
		\exists \pi^\ast\in\Stationary\text{ s.t. }
		\sup_{\pi\in\Pi_B} \mathbf P_s^\pi(\reach) = \mathbf P_s^{\pi^\ast}(\reach) , \quad \text{ for all } s\in \set{1,\ldots,n}.
	\end{equation*}
\end{problem}

In the following, we shall assume that the CTMDPs and all bounds in the above decision
problems are given using rational numbers. That is, rates of transitions in each generator
matrix is a rational number, and the time bound $B$ is a rational number. 

\begin{theorem}[\cite{Buchholz2011,Miller_68}]
\label{thm:optimal_policy}
A policy $\pi\in\Pi_B$ is optimal if $\dv_t$, the decision vector taken by $\pi$ at time $B-t$, maximizes for almost all $t\in[0,B]$
\begin{equation}
\label{eq:optimal_policy}
\max_{\dv_t}(\Qb^{\dv_t} W_t^\pi) \text{ with } \frac{d}{dt} W_t^\pi = \Qb^{\dv_t}W_t^\pi,
\end{equation}
with the initial condition $W_0^\pi(\good) =1$ and $W_0^\pi(s) =0$ for all $s\in\{1,2,\ldots,n\}$.
There exists a piecewise constant policy $\pi$ that maximizes the equations.
\end{theorem}

The maximization in Equation~\eqref{eq:optimal_policy} above is performed element-wise. 
Equation~\eqref{eq:optimal_policy} should be solved forward in time to construct the policy $\pi$ backward 
in time due to the definition $\dv_t = \pi_{B-t}$. 
One can alternatively write down \eqref{eq:optimal_policy} directly backward in time based on $\pi_t$.

The proof of Theorem~\ref{thm:optimal_policy} is constructive \cite{Buchholz2011,Miller_68} and is based on the following sets for any vector $W$:
\begin{align}
& \mathcal F_1(W) = \set{\dv\in\mathcal D\,|\, \dv \text{ maximizes } \Qb^{\dv} W}, \nonumber\\
&  \mathcal F_2(W) = \set{\dv\in  \mathcal F_1(W)  \,|\, \dv \text{ maximizes } [\Qb^{\dv}]^2 W},\label{eq:lexi_set}\\
& \cdots\nonumber\\
&  \mathcal F_j(W) = \set{\dv\in  \mathcal F_{j-1}(W)  \,|\, \dv \text{ maximizes } [\Qb^{\dv}]^j W}.\nonumber
\end{align}
The sets $\mathcal F_j(W)$ form a sequence of decreasing sets such that 
$\mathcal F_1(W)\supseteq \mathcal F_2(W)\supseteq\ldots \supseteq \mathcal F_{n+2}(W) = \mathcal F_{n+k}(W)$ for all $k>2$. 
An optimal piecewise constant policy is the one that satisfies the condition $\dv_t\in \mathcal F_{n+2}(W_t^\pi)$ for all $t\in[0,B]$. 
Note that if $\mathcal F_j(W_t^\pi)$ has only one element for some $j$, $\mathcal F_{k}(W_t^\pi) = \mathcal F_j(W_t^\pi)$ 
for all $k\ge j$ and that element is the optimal decision vector. 
The next proposition shows that when $\mathcal F_{n+2}(W_t^\pi)$ has more than one element, we can pick any one (and in fact, switch between them arbitrarily).

\begin{proposition}\label{prop:multi_dec}
	Let $\pi$ be an optimal policy satisfying Equation~\eqref{eq:optimal_policy}.
	Take any $t^\ast$ such that $\mathcal F_{n+2}(W_{t^\ast}^\pi)\neq \lim_{t \uparrow t^\ast}F_{n+2}(W_{t}^\pi)$. 
	If $\mathcal F_{n+2}(W_{t^\ast}^\pi)=\{\dv^1,\dv^2,\dots,\dv^p\}$ for some $p>1$ and 
	\begin{equation*}
	\Delta_i := \sup\set{\delta>0\,|\, \dv^i\in \mathcal F_{n+2}(W_{t}^{\pi})\text{ for all } t\in[t^\ast,t^\ast+\delta)},\quad \forall i\in\{1,2,\ldots,p\}
	\end{equation*}
	Then, $\Delta_1=\Delta_2=\dots=\Delta_p$.\\
%
%
	Suppose there are points $\delta_1, \delta_2$ such that $t^\ast\leq \delta_1 < \delta_2 <  t^\ast+\Delta_1$ and
	for all $t \in [\delta_1,\delta_2)$, we have $\pi_{B-t} = \dv$ for some $\dv\in\mathcal F_{n+1}(W_{t^\ast}^\pi)$. 
	If $\pi'$ is a policy that agrees with $\pi$ on $[0, \delta_1)$ but 
	for all $t \in [\delta_1, \delta_2)$, we have $\pi'_{B - t} = \dv'$ for some $\dv' \in \mathcal F_{n+1}(W_{t^\ast}^\pi)\setminus \set{\dv}$,
	then $\pi'$ also satisfies Equation~\eqref{eq:optimal_policy} for almost all $t\in [0, \delta_2)$.
\end{proposition}

\begin{proof}
	Since $\mathcal F_{n+2}(W_{t^\ast}^\pi)=\mathcal F_{n+k}(W_{t^\ast}^\pi)$ for all $k>2$, for any $\dv^i$ and $\dv^j$ belonging to the 
	set $\mathcal F_{n+2}(W_{t^\ast}^\pi)$, we have $[\Qb^{\dv^i}]^l W_{t^\ast}^\pi=[\Qb^{\dv^j}]^l W_{t^\ast}^\pi$ for all $l\geq0$.
	Pick $\delta>0$ sufficiently small such that $\{\dv^1,\dv^2,\dots,\dv^p\}\subseteq \mathcal F_{n+2}(W_{t}^\pi)$ for all $t\in[t^\ast,t^\ast+\delta)$.
	 If the policy $\pi$ selects $\dv^i$ for all $t\in[t^\ast,t^\ast+\delta)$, we can write
	\begin{equation*}
		W_{t}^\pi=e^{[\Qb^{\dv^i}](t-t^\ast)}W_{t^\ast}^\pi \text{ for } t\in[t^\ast,t^\ast+\delta),
	\end{equation*}
	where $e^{\Gamma}:=\sum_{k=0}^\infty \frac{1}{k!}\Gamma^k$ denotes the exponential of a matrix $\Gamma$. Therefore, using the fact that $[\Qb^{\dv^i}]^l W_{t^\ast}^\pi=[\Qb^{\dv^j}]^l W_{t^\ast}^\pi$ for all $l\geq0$ we have
	 \begin{equation}
	 \label{eq:multi_dec_equal_exponents}
	 	e^{[\Qb^{\dv^i}](t-t^\ast)}W_{t^\ast}^\pi=e^{[\Qb^{\dv^j}](t-t^\ast)}W_{t^\ast}^\pi, \quad \forall\, t\ge t^\ast.
	 \end{equation}
	 Similarly, we have
	\begin{equation}
		\label{eq:multi_dec_equal_derivitives}
		[\Qb^{\dv^i}]^l e^{[\Qb^{\dv^i}]\Delta}W_{t^\ast}^\pi=[\Qb^{\dv^j}]^l e^{[\Qb^{\dv^j}]\Delta}W_{t^\ast}^\pi,\quad \forall\, l\ge 0 \text{ and }\Delta\ge 0.
	\end{equation}
	Now take any $i = \arg\min_j\Delta_j$, thus $\Delta_i\le \Delta_j$ for all $j$. Also take $\dv'\in \mathcal F_{n+2}(W_{t^\ast+\Delta_i}^\pi)$ and $\dv'\neq \dv^i$
	(this is possible due to the definition of $\Delta_i$).
	Denote by $h$ the smallest integer for which $1\leq h \leq n+2$ and
	 \begin{equation*}
	 [\Qb^{\dv'}]^h W_{t^\ast+\Delta_i}^\pi>[\Qb^{\dv^i}]^h W_{t^\ast+\Delta_i}^\pi\Rightarrow [\Qb^{\dv'}]^h e^{[\Qb^{\dv^i}]\Delta_i}W_{t^\ast}^\pi>[\Qb^{\dv^i}]^h e^{[\Qb^{\dv^i}]\Delta_i}W_{t^\ast}^\pi.
	 \end{equation*}
	 Combining the above expression with Equation~\eqref{eq:multi_dec_equal_derivitives}, we get 
	 \begin{equation*}
	 [\Qb^{\dv'}]^h e^{[\Qb^{\dv^i}]\Delta_i}W_{t^\ast}^\pi>[\Qb^{\dv^j}]^h e^{[\Qb^{\dv^j}]\Delta_i}W_{t^\ast}^\pi\,\,\Rightarrow\,\,  [\Qb^{\dv'}]^h W_{t^\ast+\Delta_i}^\pi>[\Qb^{\dv^j}]^h W_{t^\ast+\Delta_i}^\pi,
	 \end{equation*}
	 which implies that $\Delta_j\le \Delta_i$ for any $j$. The particular selection of $i$ results in $\Delta_j = \Delta_i$ for all $i,j$.
	 The second part of the proposition is obtained by setting $\Delta = (\delta_2 - \delta_1)$ in Equation~\eqref{eq:multi_dec_equal_derivitives} and using the definition of the exponential of a matrix.
\end{proof}

The above proposition highlights the fact that whenever $\mathcal F_{n+2}(W_{t}^\pi)$ contains more than one decision vector over a time interval, 
one can construct infinitely many optimal policies by arbitrarily switching between such decision vectors. 
In the rest of this paper, we restrict our attention to optimal policies that take only mandatory switches: the 
optimal policy will take an element of $\mathcal F_{n+2}(W_{t}^\pi)$ as long as possible. 
This does not influence Problems~\ref{prob:prob_main} and \ref{prob:prob_switch}.

The major challenge in the computation of the optimal policy, thus answering the reachability problem, is the computation of the largest 
time $t\in[0,B)$ such that $ \mathcal F_{n+2}(W_t^\pi)\ne \mathcal F_{n+2}(W_{t^-}^\pi)$, where $W_{t^-}^\pi$ 
denotes the value of $W_{t-\delta}^\pi$ with $\delta$ converging to zero from the right. 
Suppose a decision vector $\dv_0\in \mathcal F_{n+2}(W_0^\pi)$ is selected. The optimal policy will change at the following time point:
\begin{equation*}
	t'' := \sup\set{t\,|\, \dv_0\in \mathcal F_{n+2}(W_{t'}^\pi)\text{ for all } t'\in[0,t)}.
\end{equation*}

\section{Conditional Decidability of Problems~\ref{prob:prob_main} and~\ref{prob:prob_switch}}\label{sec:decidability}

\subsection{Schanuel's Conjecture and its Implications}


Our decidability results will assume Schanuel's Conjecture for the complex numbers,
a unifying conjecture in transcendental number theory
(see, e.g., \cite{Lang1966}).
Recall that a \emph{transcendence basis} of a field extension $L/K$ is a subset 
$S \subseteq L$ such that $S$ is algebraically independent over $K$ and $L$ is algebraic over $K(S)$. 
The \emph{transcendence degree} of $L/K$ is the (unique) cardinality of some basis.

\begin{conjecture}[\textbf{Schanuel's Conjecture (SC)}]
Let $a_1,\ldots, a_n$ be complex numbers that are linearly independent over rational numbers $\rats$. 
Then the field $\rats(a_1,\ldots,a_n,e^{a_1},\ldots,e^{a_n})$ has transcendence degree at least $n$ over $\rats$.
\end{conjecture}

An important consequence of Schanuel's conjecture is that the theory of reals $(\reals, 0, 1, +, \cdot, \leq)$
remains decidable when extended with the exponential and trigonometric functions over bounded domains.\footnote{
	We note that while the result is claimed in several papers \cite{MACINTYRE2008,Macintyre2016}, 
	a complete proof of this result has never been published.
	Thus, it would be nice to have a ``direct'' proof of our main theorem (Theorem~\ref{th:main-theorem}) starting with Schanuel's conjecture. 
	We do not know such a proof.
}

\begin{theorem}[Macintyre and Wilkie (see \cite{MACINTYRE2008,Macintyre2016})]
\label{th:macintyre}
Assume $\textbf{SC}$.
For any $n\in \nats$, the theory $\theory :=(\reals, \exp\upharpoonright [0,n], \sin\upharpoonright [0,n], \cos\upharpoonright [0,n])$ is decidable.
\end{theorem}

Our main result will show that Problems~\ref{prob:prob_main} and~\ref{prob:prob_switch} can be decided based on Theorem~\ref{th:macintyre}.
In fact, Problem~\ref{prob:prob_switch} can be decided directly from Schanuel's conjecture and recent results on exponential polynomials \cite{joel2016}.
\begin{theorem}[Main Result]
\label{th:main}
Assume $\textbf{SC}$.
Then Problems~\ref{prob:prob_main} and~\ref{prob:prob_switch} are decidable.
\end{theorem}

In contrast, solving the time-bounded reachability problem for \emph{stationary} policies is decidable unconditionally.
This is because fixing a stationary policy reduces the time-bounded reachability problem to one on CTMCs, and one can
use the decision procedure from \cite{Aziz2000}. 

\subsection{Non-tangential Zeros}

Recall that the solution to a first-order linear ODE of dimension $n$:
\[
\frac{d}{dt}X_t = A X_t, \quad z_t = C X_t
\]
with real matrices $A$ and $C$ and 
real initial condition $X_0\in \reals^n$, can be written as 
$z_t = C e^{At} X_0$
where $e^{\Gamma}$ denotes the exponential of a square matrix $\Gamma$,
and defined as the infinite sum $e^{\Gamma}:=\sum_{k=0}^\infty \frac{1}{k!}\Gamma^k$ that is guaranteed to converge for any matrix $\Gamma$. 
The function can be expressed as an exponential polynomial
$z_t = \sum_{j=1}^k P_t(j)e^{\lambda_jt}$,
where $\lambda_1,\ldots, \lambda_k$ are the distinct (real or complex) eigenvalues of $A$. 
Each function $P_t(j)$ is a polynomial
function of $t$ possibly with complex coefficients and has a degree one less than the multiplicity of the eigenvalue $\lambda_j$.
Since the eigenvalues come in conjugate pairs, we can write the real-valued function $z$ as
\begin{equation}
\label{eq:closed_form}
z_t = \sum_{j=1}^k e^{a_jt} \sum_{l=0}^{m_j - 1} c_{j,l} t^l \cos(b_jt + \varphi_{j,l}),
\end{equation}
where the eigenvalues are $a_j \pm \mathbf{i} b_j$ with multiplicity $m_j$. 
If $A$, $X_0$, and $C$ are over the rational numbers, then $a_j$, $b_j$, $c_{j,l}$ are real algebraic 
and $\varphi_{j,l}$ is such that $e^{\mathbf{i}\varphi_{j,l}}$ is algebraic for all $j$ and $l$.
We can symbolically compute derivatives of $z$ which also become functions with a similar closed-form as in \eqref{eq:closed_form}.


\begin{definition}
\label{def:tangential}
The function $z_t$ has a \emph{zero} at $t = t^*$ if $z_{t^*} = 0$.
The zero is said to be \emph{non-tangential} if there is an $\varepsilon>0$ such that $z_{t_1}z_{t_2}<0$ for all $t_1\in(t^\ast-\varepsilon,t^\ast)$ and all $t_2\in(t^\ast,t^\ast+\varepsilon)$.
The zero is called \emph{tangential} if there is an $\varepsilon>0$ such that $z_{t_1}z_{t_2}>0$ for all $t_1\in(t^\ast-\varepsilon,t^\ast)$ and all $t_2\in(t^\ast,t^\ast+\varepsilon)$.
\end{definition}
Note that there are functions with zeros that are neither tangential nor non-tangential. Consider the function $z_t = t\sin\left(\frac{1}{t}\right)$ for $t\ne 0$ and $z_0= 0$. The function does not satisfy the conditions of being tangential or non-tangential. For any $\varepsilon>0$, there are $t_1\in(-\varepsilon,0)$ and $t_2\in(0,\varepsilon)$, such that $z_{t_1}z_{t_2} = t_1t_2\sin\left(\frac{1}{t_1}\right)\sin\left(\frac{1}{t_2}\right)$ is positive. There are also $t_1$ and $t_2$ in the respective intervals that make $z_{t_1}z_{t_2}$ negative. In this paper, we only work with functions of the form \eqref{eq:closed_form} that are analytic thus infinitely differentiable. Therefore, the first non-zero derivative of $z_t$ at $t^\ast$ will decide if $t^\ast$ is tangential or not.
\begin{proposition}
\label{thm:tangential_derivative}
For any function $z_t$ of the form \eqref{eq:closed_form} such that $z_{t^\ast} = 0$ and $z\not\equiv 0$, 
there is a $k_0$ such that $\frac{d^k}{dt^k}z_t\big|_{t=t^\ast} = 0$ for all $k< k_0$ and $\frac{d^{k_0}}{dt^{k_0}}z_t\big|_{t=t^\ast} \neq 0$. 
Moreover, $t^\ast$ is tangential if $k_0$ is an even number and is non-tangential if $k_0$ is an odd number. 
\end{proposition} 
\begin{proof}
The proof is based on the Taylor series of $z_t$ at $t=t^\ast$. Take $k_0$ the order of the first non-zero derivative of $z_t$ at $t=t^\ast$. This $k_0$ always exists since otherwise $z\equiv 0$. The Taylor series of $z_t$ will be
\begin{equation}
\label{eq:Taylor}
z_t = \sum_{k=k_0}^\infty \frac{(t-t^\ast)^k}{k!}\frac{d^k}{dt^k}z_t\big|_{t=t^\ast} = (t-t^\ast)^{k_0} \frac{d^{k_0}}{dt^{k_0}}z_t\big|_{t=t^\ast}\sum_{k=0}^\infty\alpha_k(t-t^\ast)^k,
\end{equation}
for some $\{\alpha_0,\alpha_1,\ldots\}$ with $\alpha_0 = \frac{1}{k_0!}$.
%
Define the function $g$ by
$g_t:= \frac{z_t}{(t-t^\ast)^{k_0}}$ for $t\ne t^\ast$ and 
$g_{t^\ast}:=\frac{1}{k_0!}\frac{d^{k_0}}{dt^{k_0}}z_t\big|_{t=t^\ast}$.
Using \eqref{eq:Taylor}, we get that $g$ is continuous at $t^\ast$ with 
$g_{t^\ast}\neq 0.$ Therefore, there is an interval $(t^\ast-\varepsilon, t^\ast+\varepsilon)$
over which the function has the same sign as $g_{t^\ast}$. For all $t_1\in(t^\ast-\varepsilon,t^\ast)$ and $t_2\in(t^\ast,t^\ast+\varepsilon)$
\begin{align*}
&g_{t_1}g_{t^*}>0\Rightarrow \frac{z_{t_1}}{(t_1-t^\ast)^{k_0}}g_{t^*}>0 \Rightarrow (-1)^{k_0}z_{t_1}g_{t^*}>0\\
&g_{t_2}g_{t^*}>0\Rightarrow \frac{z_{t_2}}{(t_2-t^\ast)^{k_0}}g_{t^*}>0 \Rightarrow z_{t_2}g_{t^*}>0\\
&  \Rightarrow (-1)^{k_0}z_{t_1}g_{t^*} z_{t_2}g_{t^*}>0  \Rightarrow (-1)^{k_0}z_{t_1}z_{t_2}>0.
\end{align*}
This means $z_{t_1}z_{t_2}>0$ for even $k_0$ and $t^\ast$ becomes tangential, and $z_{t_1}z_{t_2}<0$ for odd $k_0$ and $t^\ast$ becomes non-tangential.
%
\end{proof}
For any function $z_t = Ce^{At}X_0$, the predicate $\mathsf{NonTangentialZero}(z,l, u)$ stating the
existence of a non-tangential zero in an interval $(l, u)$ is expressible in $\theory$:
\begin{align*}
\exists t^*\,.\,l < t^* < u \wedge z_{t^*} = 0 \wedge 
\left[\exists \varepsilon > 0\,.\, \forall t_1\in(t^*-\varepsilon,0), t_2\in(0,t^* + \varepsilon)\,.\, z_{t_1}z_{t_2} < 0\right]
\end{align*}

\subsection{Switch Points are Non-Tangential Zeroes}\label{sec:ctmdp_nontangential}

Given a CTMDP $\M$ and a piecewise constant optimal policy $\pi: [0,B]\rightarrow \mathcal D$ for the time-bounded 
reachability problem, a \emph{switch point} $t^\ast$ is a point of discontinuity of $\pi$.
Consider a switch point $t^\ast$ such that the optimal policy takes the decision vector $\dv$ in the time interval $(t^\ast-\varepsilon)$ and then switches to another decision vector $\dv'$ at time $t^\ast$ for some $\varepsilon > 0$:
\begin{align*}
& \dv\in\mathcal F_{n+2}(W_t^\pi) \text{ and } \dv'\not\in\mathcal F_{n+2}(W_t^\pi)\quad \forall t\in(t^\ast-\varepsilon, t^\ast),\\
 & \dv\not\in\mathcal F_{n+2}(W_{t}^\pi) \text{ and } \dv'\in\mathcal F_{n+2}(W_{t}^\pi)\quad \forall t\in(t^\ast, t^\ast+\varepsilon).
\end{align*}
Consider a (not necessarily unique) state $s\in S$ 
with actions $a,b \in\mathcal D_s$ such that $a\neq b$ and $\dv(s)=a$, $\dv'(s)=b$. 
Define the following set of first-order ODEs
\begin{equation}
\label{eq:systems_list}
\Sigma:\left\{
\begin{array}{lr}
\frac{d}{dt} W_t^\pi=\Qb^{\dv}W_t^\pi\\
z_t=(q^a-q^b)W_t^\pi
\end{array}\right.
\end{equation}
for $t\in (t^\ast-\varepsilon,t^\ast+\varepsilon)$, where $q^a$ and $q^b$ denote the $s^{th}$ row of the matrices $\Qb^{\dv}$ and $\Qb^{\dv'}$, respectively.
The optimal decision vector on an interval before $t^\ast$ is $\dv$, thus for all $t\in (t^\ast-\varepsilon,t^\ast)$,
$$\dv\in\mathcal F_1(W_t^\pi)\Rightarrow \Qb^\dv W_t^\pi\ge  \Qb^{\dv'}W_t^\pi\Rightarrow  (\Qb^\dv -  \Qb^{\dv'})W_t^\pi\ge 0 \Rightarrow (q^a-q^b)W_t^\pi\ge 0\Rightarrow z_t\ge 0.$$
%
The next lemma states that the switch point $t^\ast$ corresponds to a non-tangential zero for $z_t$.

\begin{lemma}
\label{lem:switch-points-non-tangential}
Let $\pi$ be an optimal piecewise constant policy for the time-bounded reachability problem with bound $B$.
Suppose $\pi(B-t) = \dv_t$ for all $t\in [0,B]$. 
Suppose that for a time point $t^\ast$, 
$\dv\in\mathcal D$ is an optimal decision before $t^\ast$ and $\dv'\neq \dv$ is optimal right after $t^\ast$. 
There is an $\varepsilon$ such that for any $s\in S$ with $\dv(s)\neq \dv'(s)$, $z_t < 0$ for all $t\in (t^\ast, t^\ast+\varepsilon)$ 
with $z_t$ defined in \eqref{eq:systems_list}.
\end{lemma}

\begin{proof}
Take $k_0$ to be the smallest index $k\le n$ with  $ \dv\not\in\mathcal F_{k+1}(W_{t^*}^\pi)$ and $\dv'\in\mathcal F_{k+1}(W_{t^*}^\pi).$ Since $\dv'$ is optimal at $t^\ast$, we have $\dv,\dv'\in\mathcal F_{k+1}(W_{t^*}^\pi)$ for all $k<k_0$. We show inductively that
\begin{equation}
\label{eq:induction_switch}
[\Qb^{\dv}]^{k+1}W_{t^*}^\pi=[\Qb^{\dv'}]^{k+1}W_{t^*}^\pi \text{ and }\frac{d^{k}}{dt^{k}}z_{t^*} = 0 \text{ for all } 0\le k< k_0.
\end{equation}
The claim is true for $k=0$:
\begin{align*}
& \dv,\dv'\in\mathcal F_1(W_{t^*}^\pi)\Rightarrow
\Qb^{\dv}W_{t^*}^\pi = \Qb^{\dv'}W_{t^*}^\pi\\
& \Rightarrow
	(\Qb^{\dv}-\Qb^{\dv'})W_{t^*}^\pi=
	\begin{bmatrix}
	\ldots\\
	 q^a-q^b \\
	 \ldots 
	\end{bmatrix}
	W_{t^*}^\pi = 0
	\Rightarrow (q^a-q^b)W_{t^*}^\pi = 0
	\Rightarrow z_{t^*} = 0.
	\end{align*}
	Now suppose \eqref{eq:induction_switch} holds for $(k-1)$ with $k<k_0$. Then
	\begin{align*}
& \dv,\dv'\in\mathcal F_{k+1}(W_{t^*}^\pi)\Rightarrow
[\Qb^{\dv}]^{k+1}W_{t^*}^\pi = [\Qb^{\dv'}]^{k+1}W_{t^*}^\pi\\
& \Rightarrow
\Qb^{\dv}[\Qb^{\dv}]^{k}W_{t^*}^\pi = \Qb^{\dv'} [\Qb^{\dv'}]^{k}W_{t^*}^\pi
\Rightarrow^{(\ast)} \Qb^{\dv}[\Qb^{\dv}]^{k}W_{t^*}^\pi = \Qb^{\dv'} [\Qb^{\dv}]^{k}W_{t^*}^\pi\\
& \Rightarrow [\Qb^{\dv}-\Qb^{\dv'}] [\Qb^{\dv}]^{k}W_{t^*}^\pi = 0
\Rightarrow^{(**)} [\Qb^{\dv}-\Qb^{\dv'}] \frac{d^{k}}{dt^{k}}X_{t^*} = 0\\
& \Rightarrow (q^a-q^b)\frac{d^{k}}{dt^{k}}X_{t^*} = 0 \Rightarrow \frac{d^{k}}{dt^{k}} z_{t^*} = 0,
	\end{align*}
where $(\ast)$ holds due to the induction assumption and $(**)$ is true due to the differential equation \eqref{eq:systems_list}. Finally, we show that $\frac{d^{k_0}}{dt^{k_0}}z_{t^*} < 0$.
\begin{align*}
& \dv\not\in\mathcal F_{k_0+1}(W_{t^*}^\pi) \text{ and } \dv'\in\mathcal F_{k_0+1}(W_{t^*}^\pi)\Rightarrow
 [\Qb^{\dv}]^{k_0+1}W_{t^*}^\pi< [\Qb^{\dv'}]^{k_0+1}W_{t^*}^\pi\\
&  \Rightarrow \Qb^{\dv}[\Qb^{\dv}]^{k_0}W_{t^*}^\pi < \Qb^{\dv'} [\Qb^{\dv'}]^{k_0}W_{t^*}^\pi
\Rightarrow^{(\imath)}\Qb^{\dv}[\Qb^{\dv}]^{k_0}W_{t^*}^\pi< \Qb^{\dv'}[\Qb^{\dv}]^{k_0}W_{t^*}^\pi\\
& \Rightarrow [\Qb^{\dv}-\Qb^{\dv'}] \frac{d^{k_0}}{dt^{k_0}}W_{t^*}^\pi<0
\Rightarrow (q^a-q^b)\frac{d^{k_0}}{dt^{k_0}}W_{t^*}^\pi<0
\Rightarrow \frac{d^{k_0}}{dt^{k_0}}z_{t^*}<0,
\end{align*}
where $(\imath)$ holds due to \eqref{eq:induction_switch} for $k_0-1$.

Since $z_{t^\ast}=0$, we can select $\varepsilon$ such that $z_{t}>0$ for all $t\in(t^\ast-\varepsilon,t^\ast)$. 
Using Taylor expansion \eqref{eq:Taylor} and 
the facts that $\frac{d^{k_0}}{dt^{k_0}}z_{t^*}<0$ and $z_t>0$ for $t\in(t^*-\varepsilon,t^\ast)$, 
we have that $k_0$ must be an odd number, which means $t^\ast$ is non-tangential 
by Prop.~\ref{thm:tangential_derivative}. The function $z_t$ changes sign from positive to negative at $t^\ast$.
\end{proof}

\subsection{Conditional Decidability}

The decision procedure for Problem~\ref{prob:prob_main} is as follows.
Fix a CTMDP $\M = (\set{1,\ldots, n} \uplus\set{\good}, \mathcal D, \Qb)$ and a bound $B$.
We inductively construct a piecewise constant optimal policy, going forward in time.
To begin, we set the initial decision vector to $\dv^1$, where $\dv^1$ 
is selected such that $\dv^1\in \mathcal F_{n+2}(W_0^\pi)$ (Equation~\eqref{eq:lexi_set}) with $W_0^\pi$ set to the indicator vector that is $1$ at the $\good$ state and $0$ in other states.

Note that in general $\mathcal F_{n+2}(W_t^\pi)$ in \eqref{eq:lexi_set} may have finitely many elements and the choice of optimal decision at time $t$, $\dv_t\in\mathcal F_{n+2}(W_t^\pi)$ is not unique. 
Based on results of Proposition~\ref{prop:multi_dec}, any arbitrary element of $\mathcal F_{n+2}(W_t^\pi)$ can be chosen; 
but, we do not alter this choice until the picked decision vector does not belong to $\mathcal F_{n+2}(W_t^\pi)$ anymore.
We know that there is a piecewise constant optimal policy $\pi$ with finitely many switches obtained from the 
charactrization in Theorem~\ref{thm:optimal_policy}. Denote the (unknown) number of switches by $k\in\nats$.

We find $k$ as follows.
We inductively check the existence of a sequence of decision vectors $\dv^1, \ldots ,\dv^k$ and
time points $t_1,\ldots, t_{k-1}$ such that the optimal 
policy (given a lexicographical order on $\mathcal D$) switches from $\dv^i$ to $\dv^{i+1}$ at time $t_i$ but does not have any switch between the time points.
Then, we check if the optimal policy
makes at least one additional switch point in the interval $(t_k, B)$.
The check reduces the question to a number of satisfiability questions in $\theory$.
If we find an additional switch, we know that the optimal strategy has at least $k+1$ switches
and continue to check if there are further switch points. If not, we know that the optimal policy has $k$ switch points.

We need some notation.
A \emph{prefix} $\sigma_k = (\dv^1, t_1, \dv^2, t_2, \ldots, t_{k-1}, \dv^k) \in (\mathcal D \times (0, B))^*\times \mathcal D$
is a finite sequence of decision vectors from $\mathcal D$ and strictly increasing time points $0 < t_1 < t_2 <\ldots< t_{k-1} < B$
such that $\dv^i \neq \dv^{i+1}$ for $i\in\set{1,\ldots, k-1}$.
Intuitively, it represents the prefix of a piecewise constant policy with the first $k-1$ switches.
For two decision vectors $\dv, \dv'$, let
$\Delta(\dv, \dv') := \set{s \mid \dv(s) \neq \dv'(s)}$ be the states at which the actions suggested by the decision vectors
differ.
For a decision vector $\dv$, let $\dv[s\mapsto b]$ denote the decision vector that maps state $s$ to action $b$ but agrees with $\dv$ otherwise.

For a prefix $\sigma_k = (\dv^1, t_1, \dv^2, t_2, \ldots, t_{k-1}, \dv^k)$, 
a state $s\in S$, and an action $b\in \mathcal D_s$,
define 
\begin{equation*}
\label{eq:def_yk}
y_t^{s,b}(\sigma_k)=\mathbf u^T(s)([\Qb^{\dv^k}]-[\Qb^{\dv^k[s\mapsto b]}])e^{[\Qb^{\dv^k}](t-t_{k-1})}e^{[\Qb^{\dv^{k-1}}](t_{k-1}-t_{k-2})}\cdots e^{[\Qb^{\dv^1}]t_1}\mathbf u(\good),
\end{equation*}
where $\mathbf u(s)$ is a vector of dimension $n+1$ 
that assigns one to $s$ and zero to every other entry.
Observe that $y_t^{s,b}(\sigma_k)$ is a solution of a set of linear ODEs similar to $z_t$ in Equation~\eqref{eq:systems_list}:
\begin{equation}
\label{eq:systems_ODE}
\left\{
\begin{array}{lr}
\frac{d}{dt} W_t=[\Qb^{\dv^k}]W_t\\
y_t^{s,b}(\sigma_k)=\mathbf u^T(s)([\Qb^{\dv^k}]-[\Qb^{\dv^k[s\mapsto b]}])W_t,
\end{array}\right.
\end{equation}
with the condition $W_{t_{k-1}} = e^{[\Qb^{\dv^{k-1}}](t_{k-1}-t_{k-2})}\cdots e^{[\Qb^{\dv^1}]t_1}\mathbf u(\good)$.
%
%

We shall use (variants of) the predicate $\mathsf{NonTangentialZero}(y^{\cdot,\cdot}, t_1,t_2)$, but write the predicates informally 
for readability.
We need two additional predicates $\mathsf{Switch}(\sigma_k, t^*, \dv')$ and $\mathsf{NoSwitch}(\sigma_{k+1})$.
The predicate $\mathsf{Switch}$ states that, given a prefix $\sigma_k$, the first switch from $\dv^k$ to a new decision vector $\dv'$ occurs at time point $t^*>t_{k-1}$.
This new switch requires three conditions.
First, there is a simultaneous non-tangential zero at $t^*$ for all dynamical systems of the form \eqref{eq:systems_ODE} associated with $y_t^{s,\dv'(s)}(\sigma_k)$, $s\in\Delta(\dv^k,\dv')$.
Second, $t^\ast$ is the first time after $t_{k-1}$ that any of the dynamical systems have a non-tangential zero.
Finally, none of the states in $S\setminus \Delta(\dv^k, \dv')$ whose action remains the same before and after the switch, have a non-tangential zero in $(t_{k-1}, t^*]$ (up to and including $t^*$):
\begin{align*}
\mathsf{Switch}& (\underbrace{(\dv^1,t_1, \ldots, t_{k-1}, \dv^k)}_{\sigma_k}, t^*, \dv') \equiv \\
& 0 < t_1 < \ldots <t_{k-1} < B \wedge (B > t^\ast>t_{k-1})\wedge(\Delta(\dv^k, \dv') \neq \emptyset) \wedge\\
& \bigwedge_{s\in \Delta(\dv^k, \dv')}  \left(
\begin{array}{c}
	\mbox{``$y_t^{s,\dv'(s)}(\sigma_k)$ has a non-tangential zero at $t^*$''} \wedge \\ 
        \mbox{``$y_t^{s,\dv'(s)}(\sigma_k)$ has no non-tangential zero in $(t_{k-1}, t^*)$''}
\end{array} \right) \wedge \\
& \bigwedge_{s\in S \setminus \Delta(\dv^k, \dv')} \mbox{``$y_t^{s,\dv'(s)}(\sigma_k)$ has no non-tangential zero in $(t_{k-1}, t^*]$''}\\
\end{align*}
The predicate $\mathsf{NoSwitch}(\sigma_{k+1})$ states that, given a prefix 
$\sigma_{k+1}$, 
the last
decision vector $\dv^{k+1}$ of $(\sigma_{k+1})$ stays optimal and does not switch to another decision vector within the interval $(t_{k}, B)$.
This is equivalent to stating that none of the dynamical systems of the from \eqref{eq:systems_ODE} associated with $y_t^{s,b}(\sigma_{k+1})$ for $s\in S, b\in \mathcal D_s\setminus\dv^{k+1}(s)$
has a non-tangential zero in $(t_{k}, B)$:
\begin{align*}
\mathsf{NoSwitch}(\sigma_{k+1}) \equiv \bigwedge_{s,b\neq \dv^{k+1}(s)} \mbox{``$y_t^{s,b}(\sigma_{k+1})$ has no non-tangential zero in $(t_{k}, B)$''}
\end{align*}
%
We can now check if the optimal strategy has exactly $k$ switches.
The first part of the predicate written
below sets up a proper $\sigma$ and the last conjunct states that there is no further
switch after the last one.
\begin{align*}
\exists & t_1,\ldots, t_k. (0 < t_1 < t_2 \ldots < t_k < B) \wedge 
 \bigwedge_{i=1}^k \mathsf{Switch}(\underbrace{\dv^1,t_1,\ldots, \dv^i, t_i, \dv^{i+1}}_{\sigma_{i+1}}) \wedge \mathsf{NoSwitch}(\sigma_{k+1}).
\end{align*}
We can enumerate these formulas with increasing $k$ over all choices of decision vectors and stop when the above formula is valid.
%
%
At this point, we know that there is a piecewise constant optimal policy with $k$ switches, which plays the decision vectors $\dv^1,\ldots, \dv^k$.
We can make one more query to check if the probability of reaching $\good$ when playing this strategy is at least a given rational vector $r\in[0,1]^{n}$:
\begin{align}
\label{eq:final_query}
\exists t_1,\ldots, t_k.(0 < t_1 &< \ldots, t_k < B)\wedge
\bigwedge_{i=1}^k \mathsf{Switch}(\dv^1,t_1,\ldots, \dv^i, t_i, \dv^{i+1}) \wedge \mathsf{NoSwitch}(\sigma_{k+1}) \nonumber\\
&  \wedge \bigwedge_{s=1}^n \mathbf u^T(s) e^{[\Qb^{\dv^{k+1}}](B - t_k)}e^{[\Qb^{\dv^{k}}](t_{k}-t_{k-1})}\cdots e^{[\Qb^{\dv^1}]t_1}\mathbf u(\good) > r(s)
\end{align}
%
This completes the proof of conditional decidability of Problem~\ref{prob:prob_main}.

\smallskip
\noindent\textbf{Conditional Decidability for Problem~\ref{prob:prob_switch}.}
A stationary policy $\dv$ is not optimal if there is a switch point.
Using the $\mathsf{Switch}$ predicate and conditional decidability of $\theory$,
this shows conditional decidability of Problem~\ref{prob:prob_switch}.

In fact, to check the presence of a single non-tangential zero, one can avoid Theorem~\ref{th:macintyre}
and get a direct construction based on Schanuel's conjecture.
This construction is similar to \cite{joel2016} and is provided in Section~\ref{app:skolem_joel}.
Unfortunately, when there are multiple switch points, we have to existentially quantify over previous 
switch points.
Thus, the techniques of \cite{joel2016} cannot be straightforwardly extended to find a direct conditional
decision procedure for Problem~\ref{prob:prob_main}.

We do not know if there is a numerical procedure that only uses an oracle for non-tangential zeros.
The problem is that, while numerical techniques can be used to bound 
each non-tangential zero with rational intervals with arbitrary precision as well as compute the
reachability probability to arbitrary precision, we do not know how to numerically detect whether
the reachability probability in~\eqref{eq:final_query} is exactly equal to a given $r$.
%
By the Lindemann-Weierstrass Theorem \cite{lang1971}, we already know that for CTMDPs with stationary optimal
strategies, the value of reachability probability for any rational time bound $B>0$ is transcendental 
and hence $\sup_{\pi\in\Pi_B} \mathbf P_s^\pi(\reach)\neq r(s)$ for all $s\in S$.
However, we cannot prove that the reachability probability remains irrational in the general case.
\section{Lower Bound: Continuous Skolem Problem}
\label{sec:skolem_ctmdp}


\begin{problem}[Bounded Continuous-Time Skolem Problem]
\label{prob:Skolem}
Given a linear ordinary differential equation (ODE)
\begin{equation}
\label{eq:ode}
\frac{d^n}{dt^n}z_t+a_{n-1}\frac{d^{n-1}}{dt^{n-1}}z_t+\cdots+a_1\frac{d}{dt}z_t + a_0 z_t=0
\end{equation}
with rational initial conditions $z_0,\frac{dz_t}{dt}|_{t=0},\ldots,\frac{d^{n-1}z_t}{dt^{n-1}}|_{t=0}\in \mathbb{Q}$
and rational coefficients $a_{n-1},a_{n-2},\ldots,a_0\in \mathbb{Q}$ and a  
time bound $B\in \mathbb{Q}$, 
the \emph{bounded continuous Skolem problem} asks whether there exists 
$0<t^*<B$ such that it is a non-tangential zero for $z_t$.
Further, we can assume w.l.o.g.\ that $z_0 = 0$ in the initial condition.\footnote{
	The assumption is w.l.o.g.\ because given a linear ODE whose solution is $z_t$, one 
	can construct another linear ODE whose solution is $y_t = t z_t$.
	Clearly, $y_0 = 0$ and there is a non-tangential zero of $z$ in $(0,B)$ iff there is
	a non-tangential zero of $y$ in $(0, B)$. 
}
\end{problem}

We note that our definition is slightly different from the usual definition of the problem, e.g., in \cite{Bell2010,joel2016},
which simply asks for any zero (i.e., $z_{t^*}=0$), not necessarily a non-tangential one.
Our version of the bounded continuous Skolem problem is also decidable assuming $\textbf{SC}$~\cite{joel2016}.
However, there is no unconditional decidability result known for this problem, even though
we only look for a non-tangential zero.

We can encode any given linear ODE of order $n$ in the form of \eqref{eq:ode} into a set of $n$ first-order linear ODE on $X:[0,B]\rightarrow \mathbb R^n$ with
\begin{equation}
\label{eq:ss_general}
\begin{cases}
\frac{d}{dt}X_t = A X_t,\quad X_0=\left[z_0,\frac{dz_t}{dt}\Big|_{t=0},\ldots,\frac{d^{n-1}z_t}{dt^{n-1}}\Big|_{t=0}\right]^T\\
z_t = C X_t,
\end{cases}
\end{equation}
with the state matrix $A$ and output matrix $C$ are
\begin{equation}
\label{eq:ss:realization}
\begin{array}{ll}
A= \begin{bmatrix}
	0& 1 & 0 & \cdots & 0 \\
	0& 0 & 1 & \cdots & 0\\
	\vdots &\vdots & \vdots & \ddots &\vdots\\
	0& 0 & 0 & \cdots & 1\\
	-a_{0}& -a_{1} & -a_{2} & \cdots & -a_{n-1}
   \end{bmatrix},
\quad
C = \begin{bmatrix} 1 & 0 & \cdots & 0 \end{bmatrix}.
\end{array}
\end{equation}
Using the representation \eqref{eq:ss_general}, the solution of the linear ODE \eqref{eq:ode} can be written as $z_t=Ce^{At}X_0.$
Therefore, the bounded continuous-time Skolem problem can be restated as whether the expression
$C e^{At} X_0$ has a non-tangential zero in the interval $(0, B)$.
%


We now reduce the bounded continuous-time Skolem problem to Problem~\ref{prob:prob_switch}.
Given an instance \eqref{eq:ss_general}-\eqref{eq:ss:realization} of the Skolem problem of dimension $n$,
we shall construct a CTMDP over states $\set{1,\ldots, 2n} \cup\set{\good, \bad}$ and bound $B$,
and just two decision vectors $\dv^a$ and $\dv^b$ that only differ in the available actions ($a$ or $b$) at state $1$.
Our reduction will ensure that the answer of the Skolem problem has a non-tangential zero iff
there is a switch in the optimal policy in the time-bounded reachability problem for bound $B$,
and thus, iff stationary policies are not optimal.

\begin{theorem}\label{th:lower_bound}
For every instance of the bounded continuous-time Skolem problem with dynamics
$\frac{d}{dt}X_t = A X_t$, $z_t = C X_t$, initial condition $X_0$, and time bound $B$,
there is a CTMDP $\M$ such that the dynamical system has a non-tangential zero in $(0, B)$
iff the optimal strategy of the CTMDP in the time-bounded reachability problem is not stationary.
\end{theorem}

We sketch the main ideas of the proof here. 
Consider the linear differential equation described by the state space representation in \eqref{eq:ss_general} 
with the initial condition $X_0$ that has its first element equal to zero $X_0(1) = 0$. 
Given the time bound $B>0$, to solve the bounded continuous Skolem problem, we are looking for the existence of a time $0<t^\ast<B$ 
such that $z_{t^\ast}=0$ is non-tangential.
Equivalently, we want to find a non-tangential zero for the function $C e^{At} X_0$, where 
$C = \begin{bmatrix} 1 & 0 & \cdots & 0 \end{bmatrix}$.

There are three obstacles to go from \eqref{eq:ss_general}  to generator matrices for a CTMDP.
Each generator matrix must have 
non-diagonal entries that are non-negative. The sum of each row of the matrix must be zero.
Moreover, the last state of the CTMDP must be absorbing.
None of these properties may hold for a general $A$.
We show a series of transformations that take the matrix $A$ to a matrix $P$ that is sub-stochastic. 
Then we construct the generator matrices of the CTMDP using $P$ that include the required absorbing state.
We denote by $\mathbf 0_{\mathsf m}$ and $\mathbf{1}_{\mathsf m}$ as row vectors of dimension $\mathsf m$ with all elements equal to zero and one, respectively.

\begin{theorem}
\label{thm:existP}
Suppose $A\in\mathbb Q^{n\times n}$, $X_0\in\mathbb Q^n$ and $C = [1,\mathbf 0_{n-1}]$ are given with $X_0(1) = 0$.
There are positive constants $\gamma,\lambda$ and a \emph{generator} matrix $P\in\mathbb Q^{(2n+1)\times(2n+1)}$ such that
\begin{equation}
\label{eq:make_CTMDP}
 Ce^{At}X_0 = \gamma e^{\lambda t} \left[C' e^{Pt}Y_0\right],\quad C' =[1,-1,\mathbf 0_{2n-1}],\quad Y_0 = [\mathbf 0_{2n},1]^T.
\end{equation}
\end{theorem}

\begin{remark}
The first equality in \eqref{eq:make_CTMDP} ensures that nature of zeros of the two functions $Ce^{At}X_0$ and $C' e^{Pt}Y_0$ are the same. 
If one of them has a non-tangential zero at $t^\ast$ the other one will also have a non-tangential zero at $t^\ast$. 
To see this, suppose $Ce^{At^\ast}X_0 = 0$ and $Ce^{At}X_0$ changes sign at $t^\ast$. 
The same things happen to $C' e^{Pt}Y_0$ due to the fact that the two functions are different with only a positive factor of $\gamma e^{\lambda t}$.     
\end{remark}

Without loss of generality, we assume the element $A_{11}$ is negative. 
This assumption is needed when constructing the CTMDP in the sequel.
If the assumption does not hold, we can always replace $A$ with $A-\lambda_0\mathbb I_n$ for a sufficiently large $\lambda_0$ and merge 
$\lambda_0$ with $\lambda$ in \eqref{eq:make_CTMDP}. 
Define the map $\phi_1:\cup_n\mathbb Q^{n\times n}\rightarrow \cup_n\mathbb Q_{\ge 0}^{2n\times 2n}$ such that $\phi_1(A)$ is obtained by replacing each entry $A_{ij}$ with the matrix $\begin{bmatrix}\alpha_{ij} &\beta_{ij}\\\beta_{ij}&\alpha_{ij}\end{bmatrix}$, 
where $\alpha_{ij}=max(A_{ij},0)$ and $\beta_{ij}=max(-A_{ij},0)$.
The map $\phi_1$ maps any square matrix to another matrix with non-negative entries (\cite{Akshay2015}).
Also define the map $\phi_2:\cup_n\mathbb Q^{n}\rightarrow \cup_n\mathbb Q^{2n}$ such that $\phi_2(X)$ replaces each entry $X(i)$ with two entries $[X(i),0]^T$.  
\begin{proposition}
\label{prop:reduc1}
We have $C'' e^{\phi_1(A)t}Y_2 = Ce^{At}X_0$ with $Y_2 := \phi_2(X_0)$ and $C'' := [1,-1,\mathbf 0_{2n-2}]$. 
\end{proposition}

\begin{proof}
We can show inductively that for any $k\in\set{0,1,2,\ldots}$, $\set{\alpha_1,\alpha_2,\ldots,\alpha_n}$, and $[\beta_1,\beta_2,\ldots, \beta_n]:= [\alpha_1,\alpha_2,\ldots, \alpha_n] A^k$, we have
\begin{align*}
& [\alpha_1,-\alpha_1,\alpha_2,-\alpha_2,\ldots, \alpha_n,-\alpha_n] \phi_1(A)^k= [\beta_1,-\beta_1,\beta_2,-\beta_2,\ldots, \beta_n,-\beta_n].
\end{align*}
Substitute $[\alpha_1,\alpha_2,\ldots,\alpha_n]$ by $C$ and $[\beta_1,\beta_2,\ldots, \beta_n]= CA^k$ to get
\begin{align*}
& C'' \phi_1(A)^k Y_2 = C''\phi_1(A)^k \phi_2(X_0) = [\beta_1,-\beta_1,\beta_2,-\beta_2,\ldots, \beta_n,-\beta_n] \phi_2(X_0)\\
& = [\beta_1,\beta_2,\ldots, \beta_n]X_0 = C A^k X_0\\
& \Rightarrow C'' e^{\phi_1(A)t}Y_2  = \sum_{k=0}^\infty \frac{t^k}{k!}C''\phi_1(A)^kY_2 = \sum_{k=0}^\infty \frac{t^k}{k!}C A^k X_0 = C e^{At} X_0.
\end{align*}
\end{proof}

Next, we define
$\lambda:=\max_i\sum_{j=1}^n|A_{ij}|+1$,
$ P_2 := \phi_1(A)-\lambda \mathbb I_n$,
and the vector $\boldsymbol\beta\in\mathbb Q^{2n}$ with
\begin{equation*}\label{eq:beta}
\boldsymbol\beta(2i-1)=\boldsymbol\beta(2i)=\max (0,-P_2Y_2(2i-1),-P_2Y_2(2i))\quad 1\leq i \leq n.
\end{equation*}
Note that the row sum of $P_2$ is at most $-1$ and $\boldsymbol\beta+P_2Y_2$ is element-wise non-negative with the maximum element
\begin{equation*}
\label{eq:gamma}
\gamma:=\max_i P_2Y_2(i) + \boldsymbol\beta(i).
\end{equation*}
\begin{proposition}
\label{prop:reduc2}
The above choices of $\lambda,\gamma$ and the matrix 
\begin{equation*}
\label{eq:P4}
P:=
\begin{bmatrix}
P_2 & \vdots & (P_2Y_2+\boldsymbol\beta)/\gamma\\
\dots & \dots & \dots \\
\mathbf{0} &\vdots & 0
\end{bmatrix}
\end{equation*}
satisfy \eqref{eq:make_CTMDP} in Theorem~\ref{thm:existP}. Moreover, $P$ is row sub-stochastic.
\end{proposition}

\begin{proof}
We can easily show by induction that
\begin{equation*}
P^k Y_0=\begin{bmatrix}
P_{2}^{k-1}(P_2Y_2+\boldsymbol\beta)/\gamma\\
0
\end{bmatrix},
\quad\forall k\in\set{1,2,\ldots}.
\end{equation*}
\begin{equation*}
C'e^{Pt}Y_0 = \sum_{k=0}^\infty \frac{t^k}{k!}C'P^k Y_0 = C'Y_0 + C'' \sum_{k=1}^\infty \frac{t^k}{k!}P_{2}^{k-1}(P_2Y_2+\boldsymbol\beta)/\gamma, 
\end{equation*}
where $C'':= [1,-1,\mathbf 0_{2n-2}]$ is the same vector as $C'$ but the last element is eliminated.
\begin{align*}
C'e^{Pt}Y_0 =  C'Y_0 +   C''e^{P_2t}Y_2/\gamma - C''Y_2/\gamma + \sum_{k=1}^\infty \frac{t^k}{k!}C''P_{2}^{k-1}\boldsymbol\beta/\gamma.
\end{align*}
The term $C'Y_0$ is zero by simple multiplication of the two vectors. $C''Y_2 = C''\phi_2(X_0) = X_0(1)$, which is also assumed to be zero. Finally, we see by induction that for all $k\in\set{1,2,\ldots}$, the elements $(2i-1)$ and $2i$ of the matrix $P_{2}^{k-1}\boldsymbol\beta$ are equal due to the particular structure of $P_2$ and $\boldsymbol\beta$. Therefore, the last sum in the above is also zero and we get
 \begin{align*}
C'e^{Pt}Y_0 & = C''e^{P_2t}Y_2/\gamma = C''e^{\phi_1(A)t-\lambda\mathbb I t}\phi_2(X_0)/\gamma \\
& = C''e^{\phi_1(A)t}\phi_2(X_0) e^{-\lambda t}/\gamma = C e^{At X_0} e^{-\lambda t}/\gamma.
\end{align*}
To show that $P$ is a sub-stochastic matrix, we recall that $P_2Y_2+\boldsymbol\beta\ge 0$ with maximum element $\gamma$. Then
\begin{equation*}
P_2\times\mathbf 1_{2n} + (P_2Y_2+\boldsymbol\beta)/\gamma\le  \phi_1(A)\mathbf 1_{2n}-\lambda\mathbf 1_{2n}+\mathbf 1_{2n} =  \phi_1(A)\mathbf 1_{2n} - \max_i\sum_{j}|A_{ij}|\le 0.
\end{equation*}
\end{proof}

As the last step, we add an additional row and column to $P$ to make it stochastic:
\begin{equation*}
\label{eq:Qa}
\Qb^a :=
\begin{bmatrix}
P_2& \vdots & \Theta &  \vdots &  (P_2Y_2+\boldsymbol\beta)/\gamma\\
\dots &  & \dots &  & \dots \\
\mathbf{0}_{2\times 2n} &\vdots &\mathbf{0}_{2\times 1} &\vdots & \mathbf{0}_{2\times 1} 
\end{bmatrix},
  \bar C=\begin{bmatrix}1&-1& \mathbf 0_{2n}\end{bmatrix}, \bar Y_0= \begin{bmatrix}\mathbf 0_{2n+1} \\ 1\end{bmatrix},
\end{equation*}
where $\Theta$ has non-negative entries and is such that $\Qb^a$ is stochastic (sum of elements of each row is zero). The added row and column correspond to an absorbing state for a CTMDP with no effect on reachability probability: $\bar C e^{t\Qb^a}\bar Y_0 = C' e^{Pt}Y_0$.
 
 Next, we obtain a second generator matrix for the CTMDP. Define $\Qb^b := \Qb^a+ K$ with
 \begin{equation*}
\label{eq:K}
K := 
\begin{bmatrix}
- r& r & \mathbf 0_{2n} \\
\mathbf 0_{(2n+1)\times 1} & \mathbf 0_{(2n+1)\times 1}   & \mathbf 0_{(2n+1)\times 2n}
\end{bmatrix},
\end{equation*}
Note that $\Qb^b$ has exactly the same transition rates as in $\Qb^a$ except the transition from state $1$ to state $2$, which is changed by $r$.

\begin{remark}
We assumed w.l.o.g. that $A_{11}$ is negative. The construction of $P_2,P,\Qb^a$ results in a positive value for $\Qb^a_{12}$. Therefore, it is possible to select both negative and positive values for $r$ such that $\Qb^b_{12} = \Qb^a_{12} + r\ge 0$.
\end{remark}

\medskip
\noindent\textbf{Construction of the CTMDP.} 
The CTMDP $\M$ has $2n+2$ states, corresponding to the rows of $\Qb^a$ and $\Qb^b$, with the absorbing state $2n+2$ associated with the $\good$ state and the absorbing state $2n+1$ with reachability probability equal to zero. 
We shall set the time bound to be $B$. $\mathcal D_s$ the set of actions that can be 
taken in state $s\in\set{2,3,\dots,2n+2}$ is singleton and $\mathcal D_1=\set{a,b}$. 
The set of decision vectors has two elements $\mathcal D = \set{\dv^a,\dv^b}$ corresponding to the actions $a,b$ taken at state $1$. For simplicity, we denote the generator matrices of these decision vectors by $\Qb^a$ and $\Qb^b$, respectively. 
Moreover, the two actions $a,b$ have the same transition rates for jumping from state $1$ to other states, 
except giving different rates $r_a, r_b$ for jumping from $1$ to $2$ such that $r_b - r_a=r$.

The optimal policy $\pi$ takes decision vector $\dv_t\in \mathcal D$ at time $B-t$ 
such that $\dv_t\in\mathcal F_{n+2}(W_t^\pi)$ for all $t\in [0,B]$ as defined in \eqref{eq:lexi_set}.
\begin{proposition}
\label{prop:selection_r}
Let $r$ have the same sign of the first non-zero element of the set $\set{\bar C\bar Y_0,\bar C\Qb^a\bar Y_0, \bar C(\Qb^a)^2\bar Y_0,\ldots}$ and such that $\Qb^a_{12} + r\ge 0$. 
This particular selection of $r$ results in the optimality of $\dv^a$ at $t=0$.
\end{proposition}

\begin{proof}
We have $W_0^\pi = \bar Y_0$ and $\mathcal F_k(W_0^\pi) =  \arg\max_{\dv} [\Qb^{\dv}]^k \bar Y_0$. Then, we need to compare $[\Qb^a]^k \bar Y_0$ with $[\Qb^b]^k \bar Y_0$ for different values of $k$ and see which one gives the first highest value. These two are the same for $k=1$ and
$\mathcal F_1(W_0^\pi) =  \arg\max_{\dv} \Qb^{\dv} \bar Y_0 = \{\dv^a,\dv^b\}$.
Suppose For $k_0>1$ is the smallest index such that $\bar C[\Qb^a]^{k_0}\bar Y_0\neq 0$. It can be shown inductively that $[{\Qb^b}]^k \bar Y_0 = [{\Qb^a}]^k \bar Y_0$ for all $1\le k\le k_0$: 
\begin{align*}
[{\Qb^b}]^k \bar Y_0 & = {\Qb^b}[{\Qb^b}]^{k-1} \bar Y_0 = (\Qb^a+K)[{\Qb^b}]^{k-1} \bar Y_0 = (\Qb^a+K)[{\Qb^a}]^{k-1} \bar Y_0 \\
& = [{\Qb^a}]^k \bar Y_0 + K[{\Qb^a}]^{k-1} \bar Y_0= 
[{\Qb^a}]^k \bar Y_0 - r \begin{bmatrix}
\bar C [{\Qb^a}]^{k-1} \bar Y_0 \\
\mathbf 0_{(2n+1)\times 2n}
\end{bmatrix}
 = [{\Qb^a}]^k \bar Y_0.
\end{align*}
This means $\mathcal F_k(W_0^\pi) =  \arg\max_{\dv} [\Qb^{\dv}]^k \bar Y_0 = \{\dv^a,\dv^b\}$ for all $1\le k\le k_0$. We have for $k=k_0+1$
\begin{align*}
[{\Qb^b}]^{k_0+1} \bar Y_0 & =
[{\Qb^a}]^{k_0+1} \bar Y_0 - r \begin{bmatrix}
\bar C [{\Qb^a}]^{k_0} \bar Y_0 \\
\mathbf 0_{(2n+1)\times 2n}.
\end{bmatrix}
\end{align*}
The first element of $[{\Qb^b}]^{k_0+1} \bar Y_0$ is strictly less than the first element of $[{\Qb^a}]^{k_0+1} \bar Y_0$ since $r$ has the same sign as $\bar C [{\Qb^a}]^{k_0} \bar Y_0$. Thus $\mathcal F_{k_0+1}(W_0^\pi) =  \arg\max_{\dv} [\Qb^{\dv}]^{k_0+1} \bar Y_0 = \{\dv^a\}$.
\end{proof}

Note that the Skolem problem is trivial with the solution $z_t=0$ for all $t\in[0,B]$ if all the elements of the set $\set{\bar C\bar Y_0,\bar C\Qb^a\bar Y_0, \bar C(\Qb^a)^2\bar Y_0,\ldots}$ are zero.

Prop.~\ref{prop:selection_r} guarantees existence of an $\varepsilon\in(0,B)$ such that $W_t^\pi$ satisfies
\begin{equation*}
\frac{d}{dt} W_t^\pi = \Qb^a W_t^\pi\quad \forall t\in(0,\varepsilon),
\end{equation*}
with the initial condition $W_0^\pi(2n+2) =1$ and $W_0^\pi(s) =0$ for all $s\in\{1,2,\ldots,2n+1\}$.

To check if the optimal policy switches to $\dv^b$ at some time point, we should check if there is $t^\ast<B$ such 
that $\dv^b\in \mathcal F_{n+2}(W_{t^*}^\pi)$. 
This is equivalent to having $t^*$ being non-tangential for the maximization in $\mathcal F_1(W_t^\pi)$, 
which means $t^*$ is non-tangential for the equation
\begin{equation*}
\Qb^a W_t^\pi =  \Qb^b W_t^\pi \Leftrightarrow K W_t^\pi = 0  \Leftrightarrow \bar C W_t^\pi = 0.
\end{equation*}
Summarizing the above derivations, we have the following set of ODEs
\begin{equation}
\label{eq:summary1}
\frac{d}{dt} W_t^\pi = \Qb^a W_t^\pi\quad, W_0^\pi = \bar Y_0,\quad z_t = \bar C W_t^\pi.
\end{equation}
The optimal policy for CTMDP $\M$ switches from $\dv^a$ to $\dv^b$ at some time point $t^\ast$ if and only if $z_t$ in \eqref{eq:summary1} has a non-tangential zero in $(0,B)$
if and only if the original dynamics $Ce^{At}X_0$ has a non-tangential zero in $(0, B)$.
This completes the proof of Theorem~\ref{th:lower_bound}.

\section{Appendix: A Direct Algorithm for Problem~\ref{prob:prob_switch}}\label{app:skolem_joel}

We now show a ``direct'' method for decidability of Problem~\ref{prob:prob_switch} based on Schanuel's conjecture but
without relying on the decidability of $\theory$.
As stated before, a switch point in a strategy corresponds to the 
existence of a non-tangential zero for the functions $y_t^{s,b}(\dv^1)$ 
for $s\in S$ and $b\in \mathcal D_s\setminus{\dv^1(s)}$.
We know $y_t^{s,b}(\dv^1)$ is an exponential polynomial of the form \eqref{eq:closed_form}.
Thus, deciding Problem~\ref{prob:prob_switch} reduces to checking if an exponential
polynomial of the form \eqref{eq:closed_form}
in one free variable $t$ has a non-tangential zero in a bounded interval.
We use the following result from~\cite{joel2016}.

\begin{theorem}[\cite{joel2016}]
\label{thm:skolem-joel}
	 Assume $\textbf{SC}$. 
	It is decidable whether an exponential polynomial of the form \eqref{eq:closed_form} has a zero in the interval $(t_1,t_2)$ with $t_1, t_2\in\rats$.
\end{theorem}

Theorem~\ref{thm:skolem-joel} decides whether a zero, not necessarily a non-tangential one, exists.
We shall use the characterization of Proposition~\ref{thm:tangential_derivative} to check if a non-tangential zero
of $y_t  := y_t^{s,b}(\dv^1)$ exists in $(0, B)$. 
Define the functions
\begin{equation}
\label{eq:power2}
z_t^k = y_t^2 + \sum_{j=1}^{k}\left(\frac{d^{j}}{dt^{j}}y_t\right)^2, \quad k\in\set{0,1,2,\ldots}.
\end{equation}
\begin{theorem}
\label{thm:application_z}
Fix rational numbers $t_1 < t_2$.
Suppose $y_t$ has a zero in the interval $(t_1,t_2)$ and $y_t$ is not identically zero over this interval. There is $k_0$ as the smallest $k$ such that $z_t^{k}$ in \eqref{eq:power2} does \emph{not} have any zero
in $(t_1, t_2)$. 
Moreover, the zero of $y_t$ in $(t_1, t_2)$ is non-tangential if $k_0$ is odd and is tangential if $k_0$ is even. 
\end{theorem}

Intuitively, the above theorem states that if $y_t$ has at least one zero in $(t_1,t_2)$, we can check for the 
existence of a tangential or non-tangential zero by a finite number of applications of Theorem~\ref{thm:skolem-joel} to functions $z_t^k$ in \eqref{eq:power2}. 
Note that $y_t$ may have both tangential and non-tangential zeros; Theorem~\ref{thm:application_z} gives a way of identifying the type of 
one of the zeros (the one with the largest order). 

\begin{proof}[Proof of Theorem~\ref{thm:application_z}]
Since $y_t$ is an exponential polynomial, so is $z_t^k$ for all $k$. 
Thus, we can use Theorem~\ref{thm:skolem-joel} to check if $z_t^k$ has a zero in $(t_1, t_2)$.
Note that $z_t^k$ is the sum of squares of $\frac{d^{j}}{dt^{j}}y_t$, which means
\begin{equation}
\label{eq:zeros_z}
z_{t^\ast}^k = 0\,\,\Rightarrow\,\, y_{t^\ast}= \frac{dy_t}{dt}\big\vert_{t=t^\ast} = \cdots = \frac{d^{k}y_t}{dt^{k}}\big\vert_{t=t^\ast}  = 0.
\end{equation}
The first part of the theorem is proved by showing that if for each $k$, $z_t^k$ has a zero in $(t_1,t_2)$, 
then $y_t$ is identically zero. 
Suppose $z_{t}^k=0$ for some $t = t^\ast_k$ in the interval $(t_1,t_2)$, for any $k\in\set{0,1,2,\ldots}$. 
Using \eqref{eq:zeros_z}, we get that $y_t=0$ for all $t\in\set{t^\ast_0,t^\ast_1,t^\ast_2,\ldots}$. 
If the set $\set{t^\ast_0,t^\ast_1,t^\ast_2,\ldots}$ is not finite, we get that $y_t$ is identically zero according to the 
identity theorem~\cite{ablowitz_fokas_2003}. 
If the set of zeros is finite, there is some $t^\ast$ that appears infinitely often in the 
sequence $(t^\ast_0,t^\ast_1,t^\ast_2,\ldots)$. 
Therefore, $z_{t^\ast}^k = 0$ for infinitely many indices, which means $\frac{d^{k}y_t}{dt^{k}}\big\vert_{t=t^\ast}  = 0$ for all $k$. 
Having $y_k$ as an analytic function, this again implies that $y_t$ is identically zero.

Since $y_t$ is not identically zero, take $k_0$ such that $z_t^{k_0}$ does not have a zero in $(t_1,t_2)$ but $z_t^{k_0-1}$ does. 
Then, there is $t^\ast\in (t_1,t_2)$ such that $y_t$ and all its derivatives up to order $k_0-1$ are zero at $t^\ast$ 
but $\frac{d^{k_0}}{dt^{k_0}}y_t\big\vert_{t=t^\ast}\ne 0$. 
This $t^\ast$ and $k_0$ satisfy the conditions of Proposition~\ref{thm:tangential_derivative}. 
Thus, $t^\ast$ is a non-tangential zero for $y_t$ if $k_0$ is odd and a tangential zero if $k_0$ is even.
\end{proof}

To check if there is a non-tangential zero in an interval $(0, B)$, we apply Theorem~\ref{thm:application_z}
to each zero of $y_t$ individually.
Suppose $y_t$ has at least one zero. We can localize all zeros of $y_t$ as follows:
\begin{enumerate}
\item Set $(t_1,t_2) := (0,B)$;
\item\label{step:return1} Set $k_0$ to be the smallest index such that $z_t^k$ in \eqref{eq:power2} does not have any zero in $(t_1,t_2)$;
\item\label{step:return2} If $k_0>0$, do the next steps:
\begin{itemize}
\item Use bisection to find an interval $(t',t'')\subset (t_1,t_2)$ such that over this interval, 
$z_t^{k_0-1}$ has a zero and $z_t^{k_0}$ and  $\frac{d^{k_0}}{dt^{k_0}}y_t$  do not have any zero;
\item Store $(t',t'')$;
\item Repeat Steps~\ref{step:return1}-\ref{step:return2} with $(t_1,t_2) := (t_1,t')$;
\item Repeat Steps~\ref{step:return1}-\ref{step:return2} with $(t_1,t_2) := (t'',t_2)$.
\end{itemize}
\end{enumerate}
The bisection used in the above algorithm sequentially splits the interval into two sub-intervals and picks the 
one that contains the zero of $z_t^{k_0-1}$. 
It stops when $\frac{d^{k_0}}{dt^{k_0}}y_t$ does not have any zero over the selected sub-interval. 
The splitting terminates after a finite number of iterations due to the fact that $\frac{d^{k_0}}{dt^{k_0}}y_t$ is 
a continuous function and non-zero at the zero of $y_t$.
The whole algorithm terminates after a finite number of iterations since $y_t$ has a finite number of zeros in $(0,B)$ 
(note that if $y_t$ has infinite number of zeros in $(0,B)$, it will be identically zero 
according to the identity theorem~\cite{ablowitz_fokas_2003}).
The output of the algorithm is a set of intervals.
Within each interval, $y_t$ has a single zero. 
Applying Theorem~\ref{thm:application_z} to each such interval will decide whether the zero is tangential or non-tangential.

	\bibliography{references}



\end{document}